\documentclass{article}
\usepackage{spconf,amsmath,graphicx}
\usepackage{microtype}
\usepackage{graphicx}
\usepackage{subfigure}
\usepackage{booktabs,pifont} 

\usepackage{empheq}

\usepackage{pgfplots}
\pgfplotsset{compat=1.12}
\usetikzlibrary{shapes,arrows}
\usetikzlibrary{positioning}
\usepackage{tikz}
\usetikzlibrary{positioning,chains,fit,shapes,calc}
\usepackage{amsmath,bm,times}
\usetikzlibrary{calc}
\usepackage{dsfont}
\usepackage{cuted}
\usepackage{caption}

\usepackage{mathtools}
\DeclarePairedDelimiter{\ceil}{\lceil}{\rceil}

\usepackage{amsthm}
\usepackage{comment}
\usepackage{enumitem}
\usepackage{arydshln}
\usepackage{cite}
\usepackage{multirow}
\usepackage{calc}
\usepackage{blkarray}
\usepackage{url}
\theoremstyle{definition}
\newtheorem{theorem}{Theorem}

\newtheorem{example}{Example}
\newtheorem{remark}{Remark}

\usepackage{graphicx}
\usepackage{dblfloatfix}
\usetikzlibrary{decorations.pathreplacing}
\usepackage{blindtext, graphicx, amsfonts,
	amssymb,multirow,epstopdf}
\def\BibTeX{{\rm B\kern-.05em{\sc i\kern-.025em b}\kern-.08em
    T\kern-.1667em\lower.7ex\hbox{E}\kern-.125emX}}

\makeatletter
\renewcommand*\env@matrix[1][*\c@MaxMatrixCols c]{%
  \hskip -\arraycolsep
  \let\@ifnextchar\new@ifnextchar
  \array{#1}}
\makeatother
\usepackage[linesnumbered,ruled]{algorithm2e}

\newcommand{\calB}{\mathcal{B}}
\newcommand{\calC}{\mathcal{C}}

\newcommand{\calG}{\mathcal{G}}

\newcommand{\calN}{\mathcal{N}}

\newcommand{\bfy}{\mathbf{y}}
\newcommand{\bfr}{\mathbf{r}}

\newcommand{\bfD}{\mathbf{D}}

\newcommand{\bfA}{\mathbf{A}}

\newcommand{\bfx}{\mathbf{x}}

\newcommand{\bfe}{\mathbf{e}}

\title{CODED MATRIX COMPUTATIONS FOR D2D-ENABLED \\ Linearized FEDERATED LEARNING}
\name{Anindya Bijoy Das$^{\dagger}$ \qquad Aditya Ramamoorthy$^{\star}$ \qquad David J. Love$^{\dagger}$ \qquad Christopher G. Brinton$^{\dagger}$}

\address{{$^{\dagger}$}School of Electrical and Computer Engineering, Purdue University, West Lafayette, IN 47907 USA\\
{$^{\star}$}Department of Electrical and Computer Engineering, Iowa State University, Ames, IA 50010 USA\\
}
\begin{document}
%
\definecolor{mygr}{rgb}{0.6,0.4,0.0}
\definecolor{my1color}{rgb}{0.6,0.4,0.0}
\definecolor{mycolor1}{rgb}{0.00000,0.44700,0.74100}%
\definecolor{mycolor2}{rgb}{0.85000,0.32500,0.09800}%
\definecolor{mycolor3}{rgb}{0.45000,0.62500,0.19800}%
\tikzset{
block/.style    = {draw, thick, rectangle, minimum height = 2em, minimum width = 2em},
sum/.style      = {draw, circle, node distance = 1cm},
sum1/.style      = {draw, circle, minimum size = 1.1 cm},
input/.style    = {coordinate},
output/.style   = {coordinate},
}
\maketitle
\begin{abstract}
Federated learning (FL) is a popular technique for training a global model on data distributed across client devices. Like other distributed training techniques, FL is susceptible to straggler (slower or failed) clients. Recent work has proposed to address this through device-to-device (D2D) offloading, which introduces privacy concerns. In this paper, we propose a novel straggler-optimal approach for coded matrix computations which can significantly reduce the communication delay and privacy issues introduced from D2D data transmissions in FL. Moreover, our proposed approach leads to a considerable improvement of the local computation speed when the generated data matrix is sparse. Numerical evaluations confirm the superiority of our proposed method over baseline approaches.
\end{abstract}
\vspace{-0.1 in}
\section{Introduction}
\vspace{-0.05 in}
\label{sec:intro}

Contemporary computing platforms are hard-pressed to support the growing demands for AI/ML model training at the network edge. While advances in hardware serve as part of the solution, the increasing complexity of data tasks and volumes of data will continue impeding scalability. In this regard, federated learning (FL) has become a popular technique for training machine learning models in a distributed manner \cite{prakash2020coded, wang2022uav, 9606848}. In FL, the edge devices carry out the local computations, and the server collects, aggregates and updates the global model.  

Recent approaches have looked at linearizing the training operations in FL \cite{prakash2020coded, dhakal2019coded}. This is advantageous as it opens the possibility for coded matrix computing techniques that can improve operating efficiency. Specifically, in distributed settings like FL, the overall job execution time is often dominated by slower (or failed) worker nodes, which are referred to as stragglers. Recently, a number of coding theory techniques \cite{lee2018speeding, dutta2016short, yu2017polynomial, das2020coded, 8849468, tandon2017gradient, dasunifiedtreatment, 8765375, 8919859,9785638} have been proposed to mitigate stragglers in distributed matrix multiplications. A toy example \cite{lee2018speeding} of such a technique for computing $\bfA^T \bfx$ across three clients is to partition $\bfA$ as $\bfA = [\bfA_0 ~|~ \bfA_1]$, and to assign them the job of computing $\bfA^T_0 \bfx$, $\bfA^T_1 \bfx$ and $\left(\bfA_0+\bfA_1\right)^T \bfx$, respectively. In a linearized FL setting, $\bfA \in \mathbb{R}^{t \times r}$ is the data matrix and $\bfx \in \mathbb{R}^{t}$ is the model parameter vector. While each client has half of the total computational load, the server can recover $\bfA^T \bfx$ if {\it any} two clients return their results, i.e., the system is resilient to one straggler. If each of $n$ clients computes $1/k_A$ fraction of the whole job of computing $\bfA^T \bfx$, the number of stragglers that the system can be resilient to is upper bounded by $n - k_A$ \cite{yu2017polynomial}. 

In contemporary edge computing systems, task offloading via device-to-device (D2D) communications has also been proposed for straggler mitigation. D2D-enabled FL has recently been studied \cite{wang2021device, tu2020network, wang2022uav}, but can add considerable communication overhead as well as compromise data privacy. In this work, we exploit matrix coding in linearized FL to mitigate these challenges. Our straggler-optimal matrix computation scheme reduces the communication delay significantly compared to the techniques in \cite{yu2017polynomial, 8849468, 8765375}. Moreover, unlike \cite{yu2017polynomial, 8849468, 8765375, 8919859, das2019random}, our scheme allows a client to access a limited fraction of matrix $\bfA$, and provides a considerable protection against information leakage. In addition, our scheme is specifically suited to sparse matrices with a significant gain in computation speed.

\vspace{-0.025 in}
\section{Network and Learning Architecture}
\vspace{-0.045 in}
\label{sec:system}
We consider a D2D-enabled FL architecture consisting of $n = k_A + s$ clients, denoted as $W_i$ for $i = 0,1, \dots, n -1$. The first $k_A$ of them are active clients (responsible for both data generation and local computation) and the next $s < k_A$ are passive clients (responsible for local computation only).

Assume that the $i$-th device has local data $(\bfD_i, \bfy_i)$, where $\bfD_i$ and $\bfy_i$ are the block-rows of full system dataset $(\bfD, \bfy)$. Under a linear regression-based ML model, the global loss function is quadratic, i.e., $f(\beta_{\ell}) = ||\bfD \beta_{\ell} - \bfy||^2$, where the model parameter after iteration $\ell$ is obtained through gradient methods as $\beta_{\ell}  = \beta_{\ell - 1} - \mu_\ell \nabla_\beta f(\beta_{\ell -1})$ and $\mu_\ell$ is the stepsize. Based on the form of $\nabla_{\beta}f(\beta_{\ell})$, the FL local model update at each device includes multiplying the local data matrix $\bfD_i$ with parameter $\beta_\ell$. For this reason, recent work has also investigated linearizing non-linear models for FL by leveraging kernel embedding techniques \cite{prakash2020coded}. Thus, our aim is to compute $\mathbf{A}^T \mathbf{x}$ -- an arbitrary matrix operation during FL training -- in a distributed fashion such that the system is resilient to $s$ stragglers. Our assumption is that any active client $W_i$ generates a block-column of matrix $\bfA$, denoted as $\bfA_i$, $i = 0, 1, \dots, k_A - 1$, such that
\vspace{-0.05in}
\begin{align}
\label{eq:disjointA}
\bfA = \begin{bmatrix} 
\bfA_0 & \bfA_1 & \dots & \bfA_{k_A - 1} 
\end{bmatrix}.
\end{align} 
\vspace{-0.25in}

\vspace{0.05in}In our approach, every client is responsible to compute the product of a coded submatrix (linear combinations of some block-columns of $\bfA$) and the vector $\bfx$. Stragglers will arise in practice from computing speed variations or failures experienced by the clients at particular times \cite{hosseinalipour2020federated, das2020coded, das2019random}. Now, similar to \cite{wang2021device, tu2020network, wagle2022embedding}, we assume that there is a set of trusted neighbor clients for every device to transmit its data via D2D communications. The passive clients receive coded submatrices only from active clients. Unlike the approaches in \cite{dhakal2019coded, yoshida2020hybrid, prakash2020coded, 9606848}, we assume that the server cannot access to any uncoded/coded local data generated in the edge devices and is only responsible for transmission of vector $\bfx$ and for decoding $\bfA^T \bfx$ once the fastest clients return the computed submatrix-vector products.  

\vspace{-0.1 in}

\vspace{-0.05 in}
\section{Homogeneous Edge Computing}
\vspace{-0.05 in}
\label{sec:homogeneous_system}
\setlength{\textfloatsep}{2pt}
\begin{algorithm}[t]
	\caption{Proposed scheme for distributed matrix-vector multiplication}
	\label{Alg:New_matvec}
   \SetKwInOut{Input}{Input}
   \SetKwInOut{Output}{Output}
   \Input{Matrix $\bfA_i$ generated in active client $i$ for $i = 0, 1, \dots, k_A - 1$, vector $\bfx$, total $n$ clients including $s < k_A$ passive clients.}
   Set weight $\, \omega_A = s+1$ \;
   Denote client $i$ as $W_i$, for $i = 0, 1, \dots, n - 1$\;
   \For{$i\gets 0$ \KwTo $k_A-1$}{
   Define $T_i = \left\lbrace i+1, \dots, i + \omega_A - 1 \right\rbrace$ (mod $k_A$)\;
   Send $\bfA_j$, where $j \in T_i$, from $W_j$ to $W_i$\;
   Client $W_i$ creates a random vector $\bfr$ of length $k_A$, computes $\tilde{\bfA}_i = \sum_{q \in T_i} r_{q} \bfA_q$ and $\tilde{\bfA}_i^T \bfx$\;
   }
   \For{$i\gets 0$ \KwTo $s-1$}{
   $W_i$ creates random vector $\tilde{\bfr}$ of size $k_A$, computes $\tilde{\bfA}_{k_A+i} = \sum_{q \in T_i} \tilde{r}_{q} \bfA_q$ and sends to  $W_{k_A + i}$\;
   Client $W_{k_A + i}$ computes $\tilde{\bfA}_{k_A + i}^T \bfx$\;
   }
   \Output{The server recovers $\bfA^T \bfx$ from the returned results by the fastest $k_A$ clients.}
\end{algorithm}

Here we assume that each active client generates equal number of columns of $\bfA$ (i.e. all $\bfA_i$'s have the same size in \eqref{eq:disjointA}) and all the clients are rated with the same computation speed. In this scenario, we propose a distributed matrix-vector multiplication scheme in Alg. \ref{Alg:New_matvec} which is resilient to any $s$ stragglers.  

The main idea is that any active client $W_j$ generates $\bfA_j$, for $0 \leq j \leq k_A - 1$ and sends it to another active client $W_i$, if $ j = i + 1, i+2, \dots, i+\omega_A - 1$ (\textrm{ modulo} $\, k_A)$. Here we set $\omega_A = s+1$, thus, any data matrix $\bfA_j$ needs to be sent to only $\omega_A - 1 = s$ other clients. Then, active client $W_j$ computes a linear combination of $\bfA_i, \bfA_{i + 1}, \dots, \bfA_{i+\omega_A - 1} \, $ (indices modulo $\, k_A$) where the coefficients are chosen randomly from a continuous distribution. Next, active client $W_i$ sends another random linear combination of the same submatrices to $W_{i+ k_A}$ (a passive client), when $i = 0, 1, \dots, s - 1$. Note that all $n$ clients receive the vector $\bfx$ from the server. Now the job of each client is to compute the product of their respective coded submatrix and the vector $\bfx$. Once the fastest $k_A$ clients finish and send their computation results to the server, it decodes $\bfA^T \bfx$ using the corresponding random coefficients. The following theorem establishes the resiliency of Alg. \ref{Alg:New_matvec} to stragglers.


\vspace{-0.02in}
\begin{theorem}
\label{thm:matvec}
Assume that a system has $n$ clients including $k_A$ active and $s$ passive clients. If we assign the jobs according to Alg. \ref{Alg:New_matvec}, we achieve resilience to any $s = n - k_A$ stragglers.
\end{theorem}

\begin{proof}
In order to recover $\bfA^T \bfx$, according to \eqref{eq:disjointA}, we need to decode all $k_A$ vector unknowns, $\bfA^T_0 \bfx, \bfA^T_1 \bfx, \dots, \bfA^T_{k_A - 1} \bfx$; we denote the set of these unknowns as $\calB$. Now we choose an arbitrary set of $k_A$ clients each of which corresponds to an equation in terms of $\omega_A$ of those $k_A$ unknowns. Denoting the set of $k_A$ equations as $\calC$,  we have $|\calB| = |\calC| = k_A$. 

Now we consider a bipartite graph $\calG = \calC \cup \calB$, where any vertex (equation) in $\calC$ is connected to some vertices (unknowns) in $\calB$ which have participated in the corresponding equation. Thus, each vertex in $\calC$ has a neighborhood of cardinality $\omega_A$ in $\calB$.
Our goal is to show that there exists a perfect matching among the vertices of $\calC$ and $\calB$. We argue this according to Hall's marriage theorem \cite{marshall1986combinatorial} for which we need to show that for any $\bar{\calC} \subseteq \calC$, the cardinality of the neighbourhood of $\bar{\calC}$, denoted as $\calN (\bar{\calC}) \subseteq \calB$, is at least as large as $|\bar{\calC}|$. Thus, for $|\bar{\calC}| = m \leq k_A$, we need to show that $|\calN (\bar{\calC})| \geq m$.

{\it Case 1}: First we consider the case that $m \leq 2s$. We assume that $m = 2p, 2p - 1$ where $1 \leq p \leq s$. Now according to Alg. \ref{Alg:New_matvec}, the participating unknowns are shifted in a cyclic manner among the equations. If we choose any $\delta$ clients out of the first $k_A$ clients $\left(W_0, W_1, W_2, \dots, W_{k_A - 1} \right)$, according to the proof of cyclic scheme in Appendix C in \cite{das2020coded}, the minimum number of total participating unknowns is $\textrm{min} (\omega_A + \delta - 1, k_A)$, where $\omega_A = s + 1$. Now according to Alg. \ref{Alg:New_matvec}, same unknowns participate in two different equations corresponding to two different clients, $W_j$ and $W_{k_A + j}$, where $j = 0, 1, \dots, s-1$. Thus, for any $|\bar{\calC}| = m = 2p, 2p - 1 \leq 2s$, we have

\vspace{-0.25in}
\begin{align*}
|\calN (\bar{\calC})| & \geq \textrm{min} \left( \omega_A + \ceil{m/2} - 1, k_A \right) \\ & = \textrm{min} \left( \omega_A + p - 1, k_A \right)  = \textrm{min} \left( s + p, k_A \right) \geq m.
\end{align*}
\vspace{-0.25in}

{\it Case 2}: Now we consider the case where $m = 2s + q$, $1 \leq q \leq k_A - 2s$. We need to find the minimum number of unknowns which participate in any set of $m$ equations. Now, the same unknowns participate in two different equations corresponding to two different clients, $W_j$ and $W_{k_A + j}$, where $j = 0, 1, \dots, s-1$. Thus, the additional $q$ equations correspond to at least $q$ additional unknowns until the total number of participating unknowns is $k_A$. Therefore, in this case

\vspace{-0.25in}
\begin{align*}
|\calN (\bar{\calC})| &\geq \textrm{min} \left( \omega_A + \ceil{2s/2} + q - 1, k_A \right) \\  = \textrm{min} & \left( \omega_A + s + q - 1, k_A \right) = \textrm{min} \left( 2s + q, k_A \right) \geq m.  
\end{align*}
\vspace{-0.25in}

\noindent Thus, for any $m \leq k_A$ (where $|\bar{\calC}| = m$), we have shown that $|\calN (\bar{\calC})| \geq |\bar{\calC}|$. So, there exists a perfect matching among the vertices of $\calC$ and $\calB$ according to Hall's marriage theorem.

Now we consider the largest matching where vertex $c_i \in \calC$ is matched to vertex $b_j \in \calB$, which indicates that $b_j$ participates in the equation corresponding to $c_i$. Let us consider a $k_A \times k_A$ system matrix where row $i$ corresponds to the equation associated to $c_i$. Now we replace this row $i$ by $\bfe_j$ which is a unit row-vector of length $k_A$ with $j$-th entry being $1$, and $0$ otherwise. Thus we have a $k_A \times k_A$ matrix where each row has only one non-zero entry which is $1$. Since we have a perfect matching, this $k_A \times k_A$ matrix has only one non-zero entry in every column. This is a permutation of the identity matrix, and thus, is full rank. Since the matrix is full rank for a choice of definite values, according to Schwartz-Zippel lemma \cite{schwartz1980fast}, it will be full rank for random choices of non-zero entries. Thus, the server can recover all $k_A$ unknowns from any $k_A$ clients, hence the system is resilient to any $s = n - k_A$ stragglers.
\end{proof}
\vspace{-0.05 in}

\vspace{-0.1in}
\begin{example}
\begin{figure}[t]
\centering
\begin{subfigure}
\centering
\resizebox{0.88\linewidth}{!}{

\definecolor{mycolor6}{rgb}{0.92941,0.69412,0.12549}%
\definecolor{mycolor7}{rgb}{0.74902,0.00000,0.74902}%
\definecolor{mycolor8}{rgb}{0.60000,0.20000,0.00000}%

\begin{tikzpicture}[auto, thick, node distance=2cm, >=triangle 45]

\draw
    
    node [sum, minimum size = 1.15cm, fill=blue!30] (blk1) {$W0$}
    node [sum, minimum size = 1.15cm,fill=blue!30,right = 1 cm of blk1] (blk2) {$W1$}
    node [sum, minimum size = 1.15cm,fill=blue!30,right = 1 cm of blk2] (blk3) {$W2$}
    node [sum, minimum size = 1.15cm,fill=blue!30,right = 3 cm of blk3] (blk4) {$W8$}
    node [sum, minimum size = 1.15cm,fill=blue!30,right = 1 cm of blk4] (blk5) {$W9$}

    node [block, fill=orange!30, minimum width = 4.8em, below = 0.5 cm of blk1] (blk11) {\Large ${\bfA}_0 $}
    node [block, fill=orange!30, minimum width = 4.8em, below = 0.5 cm of blk2] (blk21) {\Large${\bfA}_1$}
    node [block, fill=orange!30, minimum width = 4.8em, below = 0.5 cm of blk3] (blk31) {\Large${\bfA}_2$}
    node [block, fill=orange!30, minimum width = 4.8em, below = 0.5 cm of blk4] (blk41) {\Large${\bfA}_8$}
    node [block, fill=orange!30, minimum width = 4.8em, below = 0.5 cm of blk5] (blk51) {\Large${\bfA}_9$}

node at (6.4,-1) (blk00) {$\dots \dots \dots$}
node at (5,-2.5) (blk0) {\Large (a): Data generation among the active clients.}

;
\draw[->](blk1) -- node{} (blk11);
\draw[->](blk2) -- node{} (blk21);
\draw[->](blk3) -- node{} (blk31);
\draw[->](blk4) -- node{} (blk41);
\draw[->](blk5) -- node{} (blk51);

\end{tikzpicture}
}
\end{subfigure} \vspace{-0.05 in}
\begin{subfigure}
\centering
\;
\resizebox{0.88\linewidth}{!}{

\definecolor{mycolor6}{rgb}{0.92941,0.69412,0.12549}%
\definecolor{mycolor7}{rgb}{0.74902,0.00000,0.74902}%
\definecolor{mycolor8}{rgb}{0.60000,0.20000,0.00000}%

\begin{tikzpicture}[auto, thick, node distance=2cm, >=triangle 45]

\draw
 node [sum, minimum size = 1.15cm, fill=blue!30] (blk1) {$W0$}
    node [sum, minimum size = 1.15cm,fill=blue!30,right = 3.5 cm of blk1] (blk2) {$W9$}
    node [sum, minimum size = 1.15cm,fill=green!30,right = 1.7 cm of blk2] (blk_11) {$W10$}
    node [sum, minimum size = 1.15cm,fill=green!30,right = 1.7 cm of blk_11] (blk_12) {$W11$}

    node [block, fill=mycolor6!30, minimum width = 6.8em, below = 0.5 cm of blk1] (blk11) {$\left\lbrace\bfA_0, \bfA_1, \bfA_2 \right\rbrace$}
    node [block, fill=mycolor6!30, minimum width = 6.8em, below = 0.5 cm of blk2] (blk21) {$\left\lbrace\bfA_9, \bfA_0, \bfA_1 \right\rbrace$}
    node at (2.3,-1) (blk00) {$\dots \dots \dots$}
    node [block, fill=mycolor6!30, minimum width = 6.8em, below = 0.5 cm of blk_11] (blk_111) {$\left\lbrace\bfA_0, \bfA_1, \bfA_2 \right\rbrace$}
    node [block, fill=mycolor6!30, minimum width = 6.8em, below = 0.5 cm of blk_12] (blk_121) {$\left\lbrace\bfA_1, \bfA_2, \bfA_3 \right\rbrace$}
    node at (5,-2.5) (blk00) {\Large (b): Coded submatrix allocation among all the clients.}
;
\draw[->](blk1) -- node{} (blk11);
\draw[->](blk2) -- node{} (blk21);

\draw[->](blk_11) -- node{} (blk_111);
\draw[->](blk_12) -- node{} (blk_121);

\end{tikzpicture}
}
\end{subfigure}
\vspace{-0.1 in}
\caption{\small (a) Data generation and (b) submatrix allocation for $n = 12$ clients according to Alg. \ref{Alg:New_matvec} including $k_A = 10$ active and $s = 2$ passive clients. Any $\{\bfA_j, \bfA_k, \bfA_{\ell}\}$ indicates a random linear combination of the corresponding submatrices. Any $W_i$ obtains a random linear combination of $\bfA_i, \bfA_{i+1}$ and $\bfA_{i+2}$ (indices reduced mod $10$).}
\label{matvec12}
\vspace{0.05 in}
\end{figure}
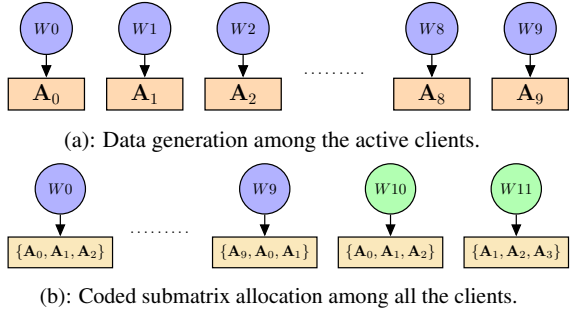
Consider a homogeneous system of $k_A = 10$ active clients and $s = 2$ passive clients. According to Alg. \ref{Alg:New_matvec}, $\omega_A = s + 1 = 3$, and client $W_i$ ($0 \leq i \leq 11$)  has a random linear combination of $\bfA_i, \bfA_{i + 1}$ and $ \bfA_{i + 2} \, \left(\textrm{indices modulo} \, 10 \right)$ as shown in Fig. \ref{matvec12}. Thus, according to Theorem \ref{thm:matvec}, this system is resilient to $s = 2$ stragglers. Note that our scheme requires any active client to send its local data matrix to only up to $s + 1 = 3$ other clients, thus involves a significantly lower communication cost in comparison to the approaches in \cite{yu2017polynomial, 8849468}.
\end{example}
\vspace{-0.05 in}

\vspace{-0.1in}
\begin{remark}
In comparison to \cite{yu2017polynomial, 8849468, 8919859}, our proposed approach is specifically suited to sparse data matrices, i.e., most of the entries of $\bfA$ are zero. The approaches in \cite{yu2017polynomial, 8849468, 8919859} assign dense linear combinations of the submatrices which can destroy the inherent sparsity of $\bfA$, leading to slower computation speed for the clients. On the other hand, our approach assigns linear combinations of limited number of submatrices which preserve the sparsity up to certain level that leads to faster computation.
\end{remark}
\vspace{-0.15 in}

\vspace{-0.05 in}
\section{Heterogeneous Edge Computing}
\vspace{-0.05 in}
\label{sec:hetero}
In this section, we extend our approach in Alg. \ref{Alg:New_matvec} to heterogenous system where the clients may have different data generation capability and different computation speeds. We assume that we have $\lambda$ different types of devices in the system, with client type $j= 0, 1, \dots, \lambda - 1$. Moreover, we assume that any active client $W_i$ generates $\alpha_i = c_{ij} \alpha$ columns of data matrix $\bfA$ and any client $W_i$ has a computation speed $\beta_i = c_{ij} \beta$, where $W_i$ is of client type $j$ and $c_{ij} \geq 1$ is an integer. Thus, a higher $c_{ij}$ indicates a ``stronger'' type client $W_i$ which can process at a $c_{ij}$ times higher computation speed than the ``weakest'' type device, where $\alpha$ is the number of the assigned columns and $\beta$ is the number of processed columns per unit time in the ``weakest'' type device. Note that $\lambda = 1$ and all $c_{ij} = 1$ lead us to the homogeneous system discussed in Sec. \ref{sec:homogeneous_system} where $0 \leq i \leq n -1$ and $j = 0$. 

Now, we have $n = k_A + s$ clients including $k_A$ active and $s$ passive clients in the heterogeneous system. Aligned to the homogeneous system, we assume that the number of passive clients of any type $j$ is less than the number of active clients of the same type. Next, without loss of generality, we sort the indices of active clients in such a way so that, $c_{i j} \geq c_{kj}$ if $i \leq k$, for $0 \leq i, k \leq k_A - 1$. We do the similar sorting for the passive clients too so that $c_{i j} \geq c_{kj}$ if $i \leq k$, for $k_A \leq i, k \leq n - 1$. Now if a client $W_i$ is of client type $j$, it requires the same time to process $c_{ij} \geq 1$ block-columns (each consisting of $\alpha$ columns) of $\bfA$  as the ``weakest'' device to process $c_{ij} = 1$ such block-column. Moreover, if it is an active client, it also generates $\alpha_i = c_{ij} \alpha$ columns of data matrix $\bfA$. Thus, client $W_i$ can be thought as a collection of $c_{ij}$ homogeneous clients of ``weakest'' types where each of the active ``weakest'' clients generates equally $\alpha$ columns of $\bfA$ and each of the ``weakest'' clients processes equally $\alpha$ columns.


\vspace{-0.01in}
\begin{theorem}
(a) A heterogeneous system of $k_A$ active and $s$ passive clients of different types can be considered as a homogeneous system of $\bar{k}_A = \sum_{i = 0}^{k_A - 1} c_{ij}$ active and $\bar{s} = \sum_{i = k_A}^{n - 1} c_{ij}$ passive clients of the ``weakest'' type. Next (b) if the jobs are assigned according to Alg. \ref{Alg:New_matvec} in the modified homogeneous system of $\bar{n} = \bar{k}_A + \bar{s}$ ``weakest'' clients, the system can be resilient to $\bar{s}$ such clients.
\end{theorem}
\vspace{-0.1in}

\begin{proof}
Each $\bfA_k$ (generated in $W_k$) in \eqref{eq:disjointA} is a block-column consisting of $c_{kj} \alpha$ columns of $\bfA$ when client $W_k$ is of client type $j$. Thus, for any $k = 0, 1, \dots, k_A - 1$, we can partition $\bfA_k$ as
$\bfA_k = \begin{bmatrix} 
\bar{\bfA}_m & \bar{\bfA}_{m+1} & \dots & \bar{\bfA}_{m + c_{kj} - 1} 
\end{bmatrix}$, where $m = \sum_{i = 0}^{k-1} c_{ij}$  and each $\bar{\bfA}_{\ell}$ is a block-column consisting of $\alpha$ columns of $\bfA$, $m \leq \ell \leq m + c_{kj} - 1$. Thus using \eqref{eq:disjointA}, we can write $\bfA = \begin{bmatrix} 
\bfA_0 & \bfA_1 & \dots & \bfA_{\bar{k}_A - 1} 
\end{bmatrix}$, where $\bar{k}_A = \sum_{i = 0}^{k_A - 1} c_{ij}$. Now from the matrix generation perspective, $k_A$ active clients in a heterogeneous system generating $\bar{k}_A$ block-columns can be considered as the same as $\bar{k}_A$ active clients in a homogeneous system generating {\it one} block-column each. 

Similarly, any client $W_i$ of type $j$ can process $c_{ij} \alpha$ columns in the same time when the ``weakest'' type device can process $\alpha$ columns. Thus, from the computation speed perspective, $k_A$ active clients and $s$ passive clients in the heterogeneous system can be thought as $\bar{k}_A = \sum_{i = 0}^{k_A-1} c_{ij}$ active clients and $\bar{s} = \sum_{i = k_A}^{n-1} c_{ij}$ passive clients, respectively, in a homogeneous system by assigning $\alpha$ coded block-columns to each client. Hence, we are done with the proof of part (a). Moreover, part (b) of the proof is straight-forward from Theorem \ref{thm:matvec} when we have $\bar{k}_A$ active and $\bar{s}$ passive clients.
\end{proof}
\vspace{-0.05in}

\begin{remark}
The heterogeneous system is resilient to $\bar{s}$ block-column processing. The number of straggler clients that the system is resilient to can vary depending on the client types.
\end{remark}
\vspace{-0.05in}
    

\begin{example}
\label{exmpl:hetero}
\begin{figure}[t]
\centering
\begin{subfigure}
\centering
\resizebox{0.83\linewidth}{!}{

\definecolor{mycolor6}{rgb}{0.92941,0.69412,0.12549}%
\definecolor{mycolor7}{rgb}{0.74902,0.00000,0.74902}%
\definecolor{mycolor8}{rgb}{0.60000,0.20000,0.00000}%

\begin{tikzpicture}[auto, thick, node distance=2cm, >=triangle 45]

\draw
    
    node [sum, minimum size = 1.15cm, fill=blue!30] (blk1) {$W2$}
    node [block,minimum height = 0.85cm, minimum width = 1.5cm,fill=blue!30] at (-2.5,0.3) (blk4) {$W1$}
    node [block, minimum height = 0.85cm, minimum width = 1cm,fill=blue!30,left = 1.3 cm of blk4] (blk5) {$W0$}
    node [sum, minimum size = 1.15cm,fill=blue!30,right = 1 cm of blk1] (blk2) {$W3$}
    node [sum, minimum size = 1.15cm,fill=blue!30,right = 1 cm of blk2] (blk3) {$W4$}

    node [block, fill=orange!30, minimum width = 4.8em, below = 0.5 cm of blk1] (blk11) {\Large$\bar{\bfA}_4 $}
    node [block, fill=orange!30, minimum width = 4.8em, below = 0.5 cm of blk2] (blk21) {\Large$\bar{\bfA}_5$}
    node [block, fill=orange!30, minimum width = 4.8em, below = 0.5 cm of blk3] (blk31) {\Large$\bar{\bfA}_6$}
    node [block, fill=orange!30, minimum width = 4.8em, below = 0.5 cm of blk4] (blk41) {\Large$\bar{\bfA}_2$}
    node [block, fill=orange!30, minimum width = 4.8em, below = 0.005 cm of blk41] (blk42) {\Large$\bar{\bfA}_3$}
    node [block, fill=orange!30, minimum width = 4.8em, below = 0.5 cm of blk5] (blk51) {\Large$\bar{\bfA}_0$}
    node [block, fill=orange!30, minimum width = 4.8em, below = 0.005 cm of blk51] (blk52) {\Large$\bar{\bfA}_1$}
    
    node at (0,-2.8) (blk0) {\large (a): Data generation among the active clients.}

;
\draw[->](blk1) -- node{} (blk11);
\draw[->](blk2) -- node{} (blk21);
\draw[->](blk3) -- node{} (blk31);
\draw[->](blk4) -- node{} (blk41);
\draw[->](blk5) -- node{} (blk51);

\end{tikzpicture}
}
\end{subfigure} 
\begin{subfigure}
\centering
\;
\resizebox{0.92\linewidth}{!}{

\definecolor{mycolor6}{rgb}{0.92941,0.69412,0.12549}%
\definecolor{mycolor7}{rgb}{0.74902,0.00000,0.74902}%
\definecolor{mycolor8}{rgb}{0.60000,0.20000,0.00000}%

\begin{tikzpicture}[auto, thick, node distance=2cm, >=triangle 45]

\draw
    
    node [sum, minimum size = 1.15cm, fill=blue!30] (bl1) {$W2$}
    node [block,minimum height = 0.85cm, minimum width = 1.5cm,fill=blue!30] at (-3,0.3) (bl4) {$W1$}
    node [block, minimum height = 0.85cm, minimum width = 1.5cm,fill=blue!30,left = 1.8 cm of bl4] (bl5) {$W0$}
    node [sum, minimum size = 1.15cm,fill=blue!30,right = 1.8 cm of bl1] (bl2) {$W3$}
    node [sum, minimum size = 1.15cm,fill=blue!30] at (-4.5,-3) (bl3) {$W4$}
    node [sum, minimum size = 1.15cm,fill=green!30,right = 2.2 cm of bl3] (bl6) {$W5$}
    node [sum, minimum size = 1.15cm,fill=green!30,right = 1.8 cm of bl6] (bl7) {$W6$}

    node [block, fill=mycolor6!30, minimum width = 4.8em, below = 0.5 cm of bl1] (bl11) {$\{ \bar{\bfA}_4, \bar{\bfA}_5, \bar{\bfA}_6 \}$}
    node [block, fill=mycolor6!30, minimum width = 4.8em, below = 0.5 cm of bl2] (bl21) {$\{ \bar{\bfA}_5, \bar{\bfA}_6, \bar{\bfA}_0 \}$}
    node [block, fill=mycolor6!30, minimum width = 4.8em, below = 0.5 cm of bl3] (bl31) {$\{ \bar{\bfA}_6, \bar{\bfA}_0, \bar{\bfA}_1 \}$}
    node [block, fill=mycolor6!30, minimum width = 4.8em, below = 0.5 cm of bl4] (bl41) {$\{ \bar{\bfA}_2, \bar{\bfA}_3, \bar{\bfA}_4 \}$}
    node [block, fill=mycolor6!30, minimum width = 4.8em, below = 0.005 cm of bl41] (bl42) {$\{ \bar{\bfA}_3, \bar{\bfA}_4, \bar{\bfA}_5 \}$}
    node [block, fill=mycolor6!30, minimum width = 4.8em, below = 0.5 cm of bl5] (bl51) {$\{ \bar{\bfA}_0, \bar{\bfA}_1, \bar{\bfA}_2 \}$}
    node [block, fill=mycolor6!30, minimum width = 4.8em, below = 0.005 cm of bl51] (bl52) {$\{ \bar{\bfA}_1, \bar{\bfA}_2, \bar{\bfA}_3 \}$}
    node [block, fill=mycolor6!30, minimum width = 4.8em, below = 0.5 cm of bl6] (bl61) {$\{ \bar{\bfA}_0, \bar{\bfA}_1, \bar{\bfA}_2 \}$}
    node [block, fill=mycolor6!30, minimum width = 4.8em, below = 0.5 cm of bl7] (bl71) {$\{ \bar{\bfA}_1, \bar{\bfA}_2, \bar{\bfA}_3 \}$}
    
     node at (-1.3,-5.5) (blk00) {\large (b): Coded submatrix allocation among all the clients.}
     
;
\draw[->](bl1) -- node{} (bl11);
\draw[->](bl2) -- node{} (bl21);
\draw[->](bl3) -- node{} (bl31);
\draw[->](bl4) -- node{} (bl41);
\draw[->](bl5) -- node{} (bl51);
\draw[->](bl6) -- node{} (bl61);
\draw[->](bl7) -- node{} (bl71);

\end{tikzpicture}
}
\end{subfigure}
\vspace{-0.1 in}
\caption{\small A heterogeneous system of $n = 7$ clients where $k_A = 5$ and $s = 2$. (a) Each of $W_0$ and $W_1$ generates $2 \alpha$ columns and each of $W_2, W_3$ and $W_4$ generates $\alpha$ columns of $\bfA \in \mathbb{R}^{t \times r}$, where $\alpha = r/7$. (b) Once the jobs are assigned, the system is resilient to stragglers.}
\label{hetero_ex}
\vspace{0.1 in}
\end{figure}
Consider the example in Fig. \ref{hetero_ex} consisting of $n = 7$ clients. There are $k_A = 5$ active clients which are responsible for data matrix generation. Let us assume, $W_0$ and $W_1$ are of type $1$ clients which generate twice as many columns of $\bfA$ than $W_2, W_3$ and $W_4$ which are of type $0$ clients. The jobs are assigned to all clients (including $s = 2$ passive clients) according to Fig. \ref{hetero_ex}(b). It can be verified that this scheme is resilient to {\it two} type $0$ clients or {\it one} type $1$ client.
\end{example}
\vspace{-0.1in}

\vspace{-0.05 in}
\section{Numerical Evaluation}
\vspace{-0.05 in}
\label{sec:num_exp}
In this section, we compare the performance of our proposed approach against different competing methods \cite{yu2017polynomial,8849468,8919859} in terms of different metrics for distributed matrix computations from the federated learning aspect. Note that the approaches in \cite{prakash2020coded ,dhakal2019coded} require the edge devices to transmit some coded columns of matrix $\bfA$ to the server which is not aligned with our assumptions. In addition, the approaches in \cite{das2020coded} and \cite{dasunifiedtreatment} do not follow the same network learning architecture as ours. Therefore, we did not include them in our comparison.

\vspace{0.05 in}{\bf Communication Delay}: We consider a homogeneous system of $n = 20$ clients each of which is a {\tt t2.small} machine in AWS (Amazon Web Services) Cluster. Here, each of $k_A = 18$ active clients generates $\bfA_i$ of size $12000 \times 1000$, thus the size of $\bfA$ is $12000 \times 18000$. The server sends the parameter vector $\bfx$ of length $12000$ to all $20$ clients including $s = 2$ passive clients. Once the preprocessing and computations are carried out according to Alg. \ref{Alg:New_matvec}, the server recovers $\bfA^T \bfx$ as soon as it receives results from the fastest $k_A = 18$ clients, thus the system is resilient to any $s = 2$ stragglers. 

\begin{table}[t]
\caption{\small Comparison among different approaches in terms of communication delay for a system with $n = 20$, $k_A = 18$ and $s = 2$.}
\vspace{-0.2 in}
\label{table:comdelay}
\begin{center}
\begin{small}
\begin{sc}
\begin{tabular}{c c c c c}
\hline
\toprule
Poly  & Ortho- & RKRP & Conv. & {\textbf{Prop.}} \\
Code  \cite{yu2017polynomial}  & Poly\cite{8849468} & Code\cite{8919859} & Code\cite{das2019random} &  {\textbf{Sch.}} \\
 \midrule
$14.13 \, s$ & $14.02 \, s$  & $2.49 \, s$ &  $2.56 \, s$ & $\mathbf{2.21 \, s}$  \\
\bottomrule
\end{tabular}
\end{sc}
\end{small}
\end{center}
\vspace{-0.22in}
\end{table}%

Table \ref{table:comdelay} shows the comparison of the corresponding communication delays (caused by data matrix transmission) among different approaches. The approaches in \cite{yu2017polynomial,8849468} require all active clients to transmit their generated submatrices to all other edge devices. Thus, they lead to much more communication delay than our proposed method which needs an edge device to transmit data to only up to $s + 1 = 3$ other devices. Note that the methods in \cite{8919859, das2019random} involve similar amounts of communication delay as ours, however, they have other limitations in terms of privacy and computation time as discussed next.

\vspace{0.05 in}{\bf Privacy}: Information leakage is introduced in FL when we consider the transmission of local data matrices to other edge devices. To protect against privacy leakage, any particular client should have access to a limited portion of the whole data matrix. Consider the heterogeneous system in example \ref{exmpl:hetero} where the clients are honest but curious. In this scenario, the approaches in \cite{yu2017polynomial,8849468,8919859, das2019random} would allow clients to access the whole matrix $\bfA$. In our approach, as shown in Fig. \ref{hetero_ex}, clients $W_0$ and $W_1$ only have access to $4/7$-th fraction of $\bfA$ and clients $W_2$, $W_3$ and $W_4$  have access to $3/7$-th fraction of $\bfA$. This provides significant protection against privacy leakage. 

\begin{table}[t]
\caption{{\small Per client product computation time where $n = 30$, $k_A = 28, s = 2$ and $\zeta = 95\%$, $98\%$ or $99\%$ entries of $\bfA$ are zero.}}
\vspace{-0.22 in}
\label{worker_comp}
\begin{center}
\begin{small}
\begin{sc}
\begin{tabular}{c c c c c}
\hline
\toprule
\multirow{2}{*}{Methods} & \multicolumn{3}{c}{Product Comp. Time (in ms)} &   \\ \cline{2-4} 
&  $\zeta = 99\%$ &  $\zeta= 98\%$ & $\zeta = 95\%$ \\
 \midrule
Poly Code  \cite{yu2017polynomial} & $54.7$ & $55.2$  & $53.7$ \\
Ortho-Poly \cite{8849468}    & $54.3$ & $54.8$  & $55.2$ \\
RKRP Code \cite{8919859}    & $55.1$ & $53.4$  & $53.7$ \\
Conv. Code \cite{das2019random}    & $56.2$ & $55.8$  & $56.8$ \\
{\textbf{Prop. Scheme}}  & $\mathbf{ 14.9}$ & $\mathbf{ 21.1}$  & $\mathbf{ 29.6}$ \\
\bottomrule
\end{tabular}
\end{sc}
\end{small}
\end{center}
\vspace{-0.08in}
\end{table}%

\vspace{0.03 in}{\bf Product Computation Time for Sparse Matrices}: Consider a system with $n = 30$ clients where $k_A = 28$ and $s = 2$. We assume that $\bfA$ is sparse, where each active client generates a sparse submatrix of size $40000 \times 1125$. We consider three different scenarios with three different sparsity levels for $\bfA$ where randomly chosen $95\%$, $98\%$ and $99\%$ entries of $\bfA$ are zero. Now we compare our proposed Alg. \ref{Alg:New_matvec} against different methods in terms of per client product computation time (the required time for a client to compute its assigned submatrix-vector product) in Table \ref{worker_comp}. The methods in \cite{yu2017polynomial, 8849468, 8919859, das2019random} assign linear combinations of $k_A = 28$ submatrices to the clients. Hence, the inherent sparsity of $\bfA$ is destroyed in the encoded submatrices. On the other hand, our approach combines  only $s +1 =3$ submatrices to obtain the coded submatrices. Thus, the clients require a significantly less amount of time to finish the respective tasks in comparison to \cite{yu2017polynomial, 8849468, 8919859, das2019random}.




\vfill\pagebreak

\bibliographystyle{IEEEbib}
\bibliography{citations}

\begin{thebibliography}{10}

\bibitem{prakash2020coded}
Saurav Prakash, Sagar Dhakal, Mustafa~Riza Akdeniz, Yair Yona, Shilpa Talwar,
  Salman Avestimehr, and Nageen Himayat,
\newblock ``Coded computing for low-latency federated learning over wireless
  edge networks,''
\newblock {\em IEEE Jour. on Sel. Areas in Comm.}, vol. 39, no. 1, pp.
  233--250, 2020.

\bibitem{wang2022uav}
Su~Wang, Seyyedali Hosseinalipour, Maria Gorlatova, Christopher~G Brinton, and
  Mung Chiang,
\newblock ``Uav-assisted online machine learning over multi-tiered networks: A
  hierarchical nested personalized federated learning approach,''
\newblock {\em IEEE Trans. on Net. and Serv. Manag.}, 2022.

\bibitem{9606848}
Jer~Shyuan Ng, Wei Yang~Bryan Lim, Zehui Xiong, Xianbin Cao, Dusit Niyato,
  Cyril Leung, and Dong~In Kim,
\newblock ``A hierarchical incentive design toward motivating participation in
  coded federated learning,''
\newblock {\em {IEEE} J. Sel. Areas Commun.}, vol. 40, no. 1, pp. 359--375,
  2022.

\bibitem{dhakal2019coded}
Sagar Dhakal, Saurav Prakash, Yair Yona, Shilpa Talwar, and Nageen Himayat,
\newblock ``Coded federated learning,''
\newblock in {\em IEEE Globecom Workshop}, 2019, pp. 1--6.

\bibitem{lee2018speeding}
Kangwook Lee, Maximilian Lam, Ramtin Pedarsani, Dimitris Papailiopoulos, and
  Kannan Ramchandran,
\newblock ``Speeding up distributed machine learning using codes,''
\newblock {\em IEEE Trans. on Info. Th.}, vol. 64, no. 3, pp. 1514--1529, 2018.

\bibitem{dutta2016short}
Sanghamitra Dutta, Viveck Cadambe, and Pulkit Grover,
\newblock ``Short-dot: Computing large linear transforms distributedly using
  coded short dot products,''
\newblock in {\em Proc. of Adv. in Neur. Inf. Proc. Syst.}, 2016, pp.
  2100--2108.

\bibitem{yu2017polynomial}
Qian Yu, Mohammad Maddah-Ali, and Salman Avestimehr,
\newblock ``Polynomial codes: an optimal design for high-dimensional coded
  matrix multiplication,''
\newblock in {\em Proc. of Adv. in Neur. Inf. Proc. Syst.}, 2017, pp.
  4403--4413.

\bibitem{das2020coded}
Anindya~Bijoy Das and Aditya Ramamoorthy,
\newblock ``Coded sparse matrix computation schemes that leverage partial
  stragglers,''
\newblock {\em IEEE Trans. on Info. Th.}, vol. 68, no. 6, pp. 4156--4181, 2022.

\bibitem{8849468}
M.~Fahim and V.~R. Cadambe,
\newblock ``Numerically stable polynomially coded computing,''
\newblock {\em IEEE Trans. on Info. Th.}, vol. 67, no. 5, pp. 2758--2785, 2021.

\bibitem{tandon2017gradient}
Rashish Tandon, Qi~Lei, Alexandros~G Dimakis, and Nikos Karampatziakis,
\newblock ``Gradient coding: Avoiding stragglers in distributed learning,''
\newblock in {\em Proc. of Intl. Conf. on Mach. Learn.}, 2017, pp. 3368--3376.

\bibitem{dasunifiedtreatment}
Anindya~Bijoy Das and Aditya Ramamoorthy,
\newblock ``A unified treatment of partial stragglers and sparse matrices in
  coded matrix computation,''
\newblock {\em IEEE Jour. on Sel. Area. in Info. Th.}, vol. 3, no. 2, pp.
  241--256, 2022.

\bibitem{8765375}
Sanghamitra Dutta, Mohammad Fahim, Farzin Haddadpour, Haewon Jeong, Viveck
  Cadambe, and Pulkit Grover,
\newblock ``On the optimal recovery threshold of coded matrix multiplication,''
\newblock {\em IEEE Trans. on Info. Th.}, vol. 66, no. 1, pp. 278--301, 2020.

\bibitem{8919859}
A.~M. Subramaniam, A.~Heidarzadeh, and K.~R. Narayanan,
\newblock ``Random {K}hatri-{R}ao-product codes for numerically-stable
  distributed matrix multiplication,''
\newblock in {\em Proc. of Annual Conf. on Comm., Control, and Computing
  (Allerton)}, Sep. 2019, pp. 253--259.

\bibitem{9785638}
Lev Tauz and Lara Dolecek,
\newblock ``Variable coded batch matrix multiplication,''
\newblock {\em IEEE Jour. on Sel. Area. in Info. Th.}, vol. 3, no. 2, pp.
  306--320, 2022.

\bibitem{wang2021device}
Su~Wang, Mengyuan Lee, Seyyedali Hosseinalipour, Roberto Morabito, Mung Chiang,
  and Christopher~G Brinton,
\newblock ``Device sampling for heterogeneous federated learning: Theory,
  algorithms, and implementation,''
\newblock in {\em Proc. of Intl. Conf. on Comp. Comm.}, 2021, pp. 1--10.

\bibitem{tu2020network}
Yuwei Tu, Yichen Ruan, Satyavrat Wagle, Christopher~G Brinton, and Carlee
  Joe-Wong,
\newblock ``Network-aware optimization of distributed learning for fog
  computing,''
\newblock in {\em Proc. of Intl. Conf. on Comp. Comm.}, 2020, pp. 2509--2518.

\bibitem{das2019random}
Anindya~Bijoy Das, Aditya Ramamoorthy, and Namrata Vaswani,
\newblock ``Efficient and robust distributed matrix computations via
  convolutional coding,''
\newblock {\em IEEE Trans. on Info. Th.}, vol. 67, no. 9, pp. 6266--6282, 2021.

\bibitem{hosseinalipour2020federated}
Seyyedali Hosseinalipour, Christopher~G Brinton, Vaneet Aggarwal, Huaiyu Dai,
  and Mung Chiang,
\newblock ``From federated to fog learning: Distributed machine learning over
  heterogeneous wireless networks,''
\newblock {\em IEEE Comm. Mag.}, vol. 58, no. 12, pp. 41--47, 2020.

\bibitem{wagle2022embedding}
Satyavrat Wagle, Seyyedali Hosseinalipour, Naji Khosravan, Mung Chiang, and
  Christopher~G Brinton,
\newblock ``Embedding alignment for unsupervised federated learning via smart
  data exchange,''
\newblock in {\em Proc. of IEEE Glob. Comm. Conf.} IEEE, 2022, pp. 1--6.

\bibitem{yoshida2020hybrid}
Naoya Yoshida, Takayuki Nishio, Masahiro Morikura, Koji Yamamoto, and Ryo
  Yonetani,
\newblock ``Hybrid-fl for wireless networks: Cooperative learning mechanism
  using non-iid data,''
\newblock in {\em Proc. of IEEE Intl. Conf. Comm.} IEEE, 2020, pp. 1--7.

\bibitem{marshall1986combinatorial}
JR~Marshall.~Hall,
\newblock {\em Combinatorial theory},
\newblock Wiley, 1986.

\bibitem{schwartz1980fast}
Jacob~T Schwartz,
\newblock ``Fast probabilistic algorithms for verification of polynomial
  identities,''
\newblock {\em Jour. of the ACM (JACM)}, vol. 27, no. 4, pp. 701--717, 1980.

\end{thebibliography}

\end{document}